\documentclass{article}

\usepackage[round]{natbib}

\PassOptionsToPackage{numbers, compress}{natbib}
\usepackage[utf8]{inputenc} 
\usepackage[T1]{fontenc}    
\usepackage[hyperfootnotes=false]{hyperref}       
\usepackage{url}            
\usepackage{booktabs}       
\usepackage{amsfonts}       
\usepackage{nicefrac}       
\usepackage{microtype}      
\usepackage[final]{nips_2017}

\usepackage{amsmath}
\usepackage{amsfonts}
\usepackage{amssymb}
\usepackage{amsthm}
\usepackage{bbm}

\usepackage{enumitem}

\usepackage{algorithm}
\usepackage{algorithmic}

\usepackage{subfigure} 

\usepackage{tikz}
\usetikzlibrary{automata,arrows,positioning,calc}
\usepackage{pgfplots}
\usepackage{pgf}
\usetikzlibrary{positioning}
\usepgfplotslibrary{patchplots,colormaps}
\usetikzlibrary{topaths}

\theoremstyle{plain} 
\theoremstyle{plain} \newtheorem{definition}{Definition}[]
\theoremstyle{plain} 
\theoremstyle{plain} 
\theoremstyle{plain} 
\theoremstyle{plain} \newtheorem{theorem}{Theorem}[]
\theoremstyle{plain} 
\theoremstyle{plain} \newtheorem{lemma}{Lemma}[]
\theoremstyle{plain} 
\theoremstyle{plain} 
\theoremstyle{plain} \newtheorem{claim}{Claim}[]

\newcommand \defn {\mathrel{\triangleq}}

\DeclareMathOperator*{\argmax}{arg\,max}

\definecolor{mycolor1}{rgb}{0.00000,0.44700,0.74100}%
\definecolor{mycolor2}{rgb}{0.85000,0.32500,0.09800}%
\definecolor{mycolor3}{rgb}{0.92900,0.69400,0.12500}%
\definecolor{mycolor4}{rgb}{0.49400,0.18400,0.55600}%
\definecolor{mycolor5}{rgb}{0.46600,0.67400,0.18800}%

\newcommand \CA {{\mathcal{A}}}
\newcommand \CB {{\mathcal{B}}}

\newcommand \CE {{\mathcal{E}}}

\newcommand \CG {{\mathcal{G}}}

\newcommand \CL {{\mathcal{L}}}

\newcommand \CO {{\mathcal{O}}}

\newcommand \CR {{\mathcal{R}}}
\newcommand \CS {{\mathcal{S}}}
\newcommand \CT {{\mathcal{T}}}

\newcommand \CV {{\mathcal{V}}}

\newcommand \pol {\pi}

\title{Adaptive Submodular Influence Maximization with Myopic Feedback}

\author{
  Guillaume Salha\thanks{Equal Contribution} \\
  LIX, \'{E}cole Polytechnique\\
  \texttt{\small guillaume.salha@polytechnique.edu}
  \And
  Nikolaos Tziortziotis$^*$ \\
  LIX, \'{E}cole Polytechnique\\
  \texttt{ntziorzi@gmail.com} \\
  \And
  Michalis Vazirgiannis\\
  LIX, \'{E}cole Polytechnique\\
  \texttt{mvazirg@lix.polytechnique.fr}
}

\begin{document}

\maketitle

\begin{abstract}
  This paper examines the problem of adaptive influence maximization in social networks. As adaptive decision making is a time-critical task, a realistic feedback model has been considered, called myopic. In this direction, we propose the \emph{myopic adaptive greedy policy} that is guaranteed to provide a $(1 - 1/e)$-approximation of the optimal policy under a variant of the independent cascade diffusion model. This strategy maximizes an alternative utility function that has been proven to be adaptive monotone and adaptive submodular. The proposed utility function considers the cumulative number of active nodes through the time, instead of the total number of the active nodes at the end of the diffusion.  Our empirical analysis on real-world social networks reveals the benefits of the proposed myopic strategy, validating our theoretical results.
\end{abstract}

\section{Introduction}

Graphs are useful models for specifying relationships within a collection of objects.
Numerous real-life situations could be represented as nodes linked by edges, including social, biological or computer networks.
Discovering the most influential nodes in such networks has been the objective of considerable research in ML and AI communities. 
One of the most practical applications is that of \emph{product placement} or \textit{viral marketing}.
Consider a directed social network in which nodes correspond to potential customers.
If a customer owns a product then he can recommend it to his friends, according to a given diffusion model that simulates the \emph{word-of-mouth} effect.
Given a fixed budget, our objective is to select a set of customers to give a product for free, in order to \textit{maximize the spread of influence through the network}, i.e., to maximize the number of people that will finally buy this product.

Influence maximization (IM) in social networks was first studied by \citet{Domingos01}. 
\citet{kdd/KempeKT03} reformulated IM as a discrete optimization problem by introducing two diffusion models: \textit{Independent Cascade} (IC) and \textit{Linear Threshold} (LT) model.
They demonstrated that finding an optimal set of at most $k$ seed nodes, with $k$ to represent our budget, that maximizes influence in the network is NP-hard under both diffusion models.
Nevertheless, they proved that the utility function to maximize, which is the \textit{expected} number of influenced nodes, is \textit{monotone} and \textit{submodular}.
These properties in conjunction with the results of \citet{Nemhauser78a} imply that the \textit{greedy strategy} is guaranteed to be a $(1-1/e)$-approximation of the optimal set.
\citet{Feige98} highlighted that this is the best possible approximation guarantee, and considered as near-optimal \citep{Nemhauser78b, Vondrak10}.
These seminal works have inspired a large part of other research works, either to provide alternative frameworks \citep{Wang:2010,Lu:2013,AslayBBB14,kdd/HeK16,TangY16}, or to speed up the greedy algorithm via heuristics providing theoretical results \citep{Chen:2009, Goyal:2011,Borgs:2014, rossi:hal-01672970} or scalability guarantees \citep{Leskovec:2007,Jung:2012,Kim:2013}.

Most of the works on influence maximization are restricted to the \emph{non-adaptive} setting, where all seed nodes must be selected in advance.
The main drawback of this assumption is that the particular choice of seed nodes is completely driven by the diffusion model and the edge probability assignment.
Apparently, it may lead to a severe overestimation of the actual spread resulting from the chosen seed nodes \citep{Goyal:2011}.
Under this prism, we focus on the \textit{adaptive} setting of the IM problem.
Instead of selecting a number of seed nodes in advance, we select one (or more) node at a time, then we observe how its activation propagates through the network, and based on the observations made so far, we adaptively select the next seed node(s).
Actually, it constitutes a \textit{sequential decision making problem} where we should design a \textit{policy} that specifies which is the most appropriate node(s) to be selected at a given time.
It can be verified, even on small graphs, that the adaptive setting leads to higher spreads compared to the non-adaptive one, since we gradually gain more knowledge about the ground truth influence graph.

\textit{Adaptive submodularity} \citep{Golovin:2011} constitutes a natural generalization of submodularity to adaptive policies.
Similar to \cite{kdd/KempeKT03},  \citet{Golovin:2011} showed that, when the objective function under consideration is adaptive monotone and adaptive submodular, a simple adaptive greedy policy performs near-optimally.
Adaptive submodularity has been verified to be useful on several practical applications such as active learning, sensor placement, etc.
However, in the adaptive IM task, the adaptive submodularity property of the utility function holds only in the case of the unrealistic \emph{Full feedback} model.
Recently, an adaptive greedy policy has been proposed by \citet{sun18} for the adaptive multi-round IM problem where an independent diffusion is executed at each round (similar to  \emph{Full feedback}). 
\citet{Tang16} has introduced the \textit{partial-feedback model} that captures the trade-off between delay and performance. 
An $(\alpha, \beta)-$greedy policy has also been proposed that guarantees a constant approximation ratio under this model.
Nevertheless, the question of whether the adaptive submodularity property can be proved for more realistic feedback models, has not been answered yet. 

\noindent\textbf{Our contribution}
In this paper, we consider a modified version of the IC diffusion model, where an active node has several opportunities to influence its neighbors.
Moreover, we introduce a new utility function that instead of computing the number of active nodes at the end of the diffusion process, considers the cumulative number of active nodes through time.
We argue in Sec.~\ref{sec:MyopicFeedback} that these modifications are consistent with many real life applications.
The main contribution of this work is the proof that the considered utility function is adaptive monotone and adaptive submodular under the modified IC model with myopic feedback.
Therefore, the proposed \emph{myopic adaptive greedy policy} is theoretically guaranteed to reach a $(1-1/e)$-approximation ratio in terms of the expected utility of the optimal adaptive policy.
To present our theoretical analysis in a strict way, we resort to a \textit{layered graph representation}, similar to the one presented by \citet{kdd/KempeKT03}, where each one of the graph's layers illustrates the diffusion in the network at a specific time stamp.
We also prove that our two assumptions, that is i) an active node has several opportunities to influence its neighbors and ii) the active nodes cannot be deactivated through time, are necessary conditions to verify that the adaptive submodularity property of the proposed utility function is valid.
Finally, the superiority of the myopic adaptive greedy strategy over other adaptive heuristic strategies and a non-adaptive greedy strategy to the IM problem has been demonstrated on three real-life social networks.

\section{Preliminaries}
\label{sec:Preliminaries}

A social network is typically modeled as a directed graph $\mathcal{G}=(\CV,\CE)$ with each node  $v \in V$ to represent a person, and the edges $\CE \subseteq \CV \times \CV$ to reflect the relationships among them. 
To simulate the diffusion process in a social network we consider the IC model.
It is a discrete-time model where only the seed nodes are initially active.
Afterwards, each time where a node $v$ first becomes active, it has a single chance to activate/influence each of its inactive neighbors $u$, succeeding with known influence probability $p_{vu}$.
The diffusion process continues until no further activations are possible.

We consider that each edge $e \in \CE$ is associated with a particular state $o \in \CO$, with $\CO$ to be a set of possible states (whether an edge is \emph{live} or \emph{dead}). 
We denote by $\phi: \CE \rightarrow \CO$ a particular realization of the influence graph, indicating the status of edges in a particular world's state. 
It is also assumed that the realization $\Phi$ is a random variable with known probability distribution, $p(\phi) \defn \mathbb{P}[\Phi = \phi]$.

In the adaptive setting, after selecting a seed node $v \in \CV$, we get a \textit{partial observation} of the ground truth influence graph $\phi$ \citep{Golovin:2011}.
More specifically, after each step, our knowledge so far will be represented as a \textit{partial realization} $\psi \subseteq \CE \times \CO$, which is a function from a subset of $\CE$ to their states.
We use the notation $dom(\psi)$, called as domain of $\psi$, to refer to the set of nodes that are observed to be active through $\psi$.
Roughly speaking, we say that a partial realization observes an edge $e$, if some node $u \in dom(\psi)$ has revealed its status.
A partial realization $\psi$ is said to be \emph{consistent} with $\phi$, denoted by $\phi \sim \psi$, if the state of all edges observed by $\psi$ are the same in $\phi$.
Also, we say that $\psi$ is a \emph{subrealization} of $\psi'$, $\psi \subseteq \psi'$, if both of them are consistent with some $\phi$, and $dom(\psi) \subseteq dom(\psi')$.

Adaptive influence maximization constitutes a sequential decision making problem where we have to design a \textit{policy} $\pol$, determining sequentially which node(s) must be selected as seed(s) at each time step, given $\psi$.
We call as $E(\pol,\Phi) \subseteq \CV$ the seed nodes that have been selected following policy $\pol$ under realization $\phi$.
The standard IM utility function is defined as $f(\CS,\phi) \defn |\sigma(\CS,\phi))|$, with $\sigma(\CS,\phi)$ to be the  set of the influenced nodes at the end of the process under realization $\phi$, and given the seed set $\CS$.
Actually, our objective is the discovering of an optimal policy $\pol^{*}$ that maximizes the \textit{expected utility}, $f_{avg}(\pol) \defn \mathbb{E}_{\Phi}[f(E(\pol,\Phi),\Phi)]$. This can be written more concretely as: 
\begin{equation*}
\pol^* \in \argmax_{\pol} f_{avg}(\pol) \quad \text{ s.t. } ~ |E(\pol,\phi)|~\leq~k, \forall \phi.
\end{equation*}

In general, this is an NP-hard optimization problem \citep{Golovin:2011}. In the non-adaptive case, we can easily derive near-optimal policies if the utility function is monotone and submodular \citep{Nemhauser78a,kdd/KempeKT03}.
To provide generalizations of monotonicity and submodularity in such an adaptive setting,  \citet{Golovin:2011} adopt the expected marginal gain notion.
\begin{definition}
The conditional expected marginal benefit of $v \in \CV$, conditioned on partial realization $\psi$, is given as:
{\small $$ \Delta_f(v|\psi) \defn \mathbb{E}_{\Phi} \Big[ f(dom(\psi) \cup \{v\},\Phi) - f(dom(\psi),\Phi) | \Phi \sim \psi \Big].$$}
\end{definition}
This leads us to the following definitions of adaptive monotonicity and adaptive submodularity, defined w.r.t. to the distribution $p(\phi)$ over realizations.
\begin{definition}
Function $f$ is adaptive monotone {\textit iff} $\Delta_f(v|\psi) \geq 0$ for all $v \in \CV$ and $\psi$ such that $\mathbb{P}(\Phi \sim \psi) > 0$.
\end{definition}
\begin{definition}
Function $f$ is adaptive submodular {\textit iff} $\Delta_f(v|\psi) \geq \Delta(v|\psi')$, for all $v \in \CV \setminus dom(\psi')$ and $\psi \subseteq \psi'$.
\end{definition} 
Let $\pol^{\text{g}}$ be the adaptive greedy policy that given the partial realization $\psi$ selects the node $v \in \CV \setminus dom(\psi)$ with the highest expected marginal gain, $\Delta_f(v|\psi)$.
 \citet{Golovin:2011} proved that, if the utility function $f$ is adaptive monotone and adaptive submodular w.r.t. $p(\phi)$, then $\pol^{\text{g}}$ is a $(1 - 1/e)$-approximation of $\pol^*$, $f_{avg}(\pol^{\text{g}})  \geq (1 - 1/e) f_{avg}(\pol^*)$.
This constitutes a direct extension of the non-adaptive bound, which was proved to be near-optimal \citep{Nemhauser78a}.

In the adaptive IM problem, the following two concrete feedbacks can be considered:
\begin{itemize}[leftmargin=*,noitemsep,nolistsep]
\item \textit{Full-adoption feedback:} activating a seed node, we observe the {\em entire} propagation (cascade) in graph, and then we select the next seed node;
\item \textit{Myopic feedback:} activating a seed node at time $t$, we only observe the status (active or not) of the neighbors of the seed nodes at time $t+1$.
\end{itemize}
Therefore, in myopic feedback model, selecting a node at time $t$ has an impact at time $t+2$, $t+3$, and so on. 
Nevertheless, it has been shown \citep{Golovin:2011} that the standard utility function $f$ holds its adaptive submodular property \emph{only} under the \emph{full-adoption feedback} model (counterexamples are reported in \citep{Golovin:2011,VaswaniL16}).
Thus, there is no guarantee that we can discover a policy able to approximate the expected utility  of the best policy within a reasonable factor in the case of the \emph{myopic feedback model}.

\section{Myopic Feedback through Layered Graphs}
\label{sec:MyopicFeedback}
The limitations of the \textit{full-adoption feedback} (i.e., in most applications the propagation in the network is not instantaneous) motivate us to focus on the \textit{myopic feedback model} that fits better on real world.

\noindent\textbf{Utility function} To deal with this situation, we introduce an alternative utility that considers \textit{the cumulative number of active nodes over time}  instead of  the total number of active nodes at the end of the diffusion process.
More precisely, given a finite \textit{horizon} $T$, the proposed utility function is defined as:
\begin{equation*}
  \tilde{f}(\CS,\phi) \defn \sum_{t=1}^{T}|\sigma_t(\CS,\phi)|,
\end{equation*}
where $\sigma_t(\CS,\phi)$ represents the set of active nodes at time $t$ if the seed set $\CS$ has been selected under realization $\phi$.
According to $\tilde{f}$, if a node is active for three time steps, it will yield a reward equal to $3$ instead of $1$ as in the case of standard IM utility function $f$.
The proposed utility function is consistent with many real life situations.
Consider, for instance, the case of platforms with a monthly subscription, like Netflix or Amazon. Those services charge each active user every month on the date he signed up. Thus, the companies' profit increases as the users are active for longer periods. Therefore, the value of an active node is additive over time.  

\noindent\textbf{Modified IC model} Let us now introduce a slight modification of the \emph{standard} IC model, which is still consistent with most real-world applications. 
In contrast to the \emph{standard} IC model where an active node has a single chance to influence its neighbors, in the \emph{modified} IC model each active node has multiple opportunities  to influence its inactive neighbors.
In Section~\ref{sec:TheroreticalAnalysis}, we prove that the proposed utility function, $\tilde{f}$, is adaptive submodular only under the \textit{modified} IC model with myopic feedback.  

\noindent\textbf{Layered graph representation} To represent the \textit{evolution of the network over time}, we resort to a \textit{layered graph representation}, denoted as $\mathcal{G}^L$. A graph's layer corresponds to the representation of the original graph at a specific time step, with $\CL_t$ to denote the set of nodes on layer $t$.
Consider for example  the original graph illustrated at Fig.~\ref{fig:graphs}(a) and its evolution over three successive time steps.
We retrieve the same amount of information as in the case of the layered graph, Fig.~\ref{fig:graphs}(b).
Indeed, node $v$ is active at time $t$ if and only if $v_t$ is active in the layered graph.
Then, it influences its neighbor $u$ at time $t+1$ with probability $p_{vu}$.
Thus, there is a possibly \emph{live} edge from $v_t$ to $u_{t+1}$.
For the sake of simplicity, in the rest of the paper we use the next indexing $f_{\mathcal{G}}$ or $\tilde{f}_{\mathcal{G}}$ in order to explicitly declare that function  $f$ or $\tilde{f}$ is computed on graph  $\mathcal{G}$.

It can be easily verified that the two networks, the original and the layered one, are closely linked.
The following lemma highlights the fact that computing $\tilde{f}_{\mathcal{G}}$ is equivalent to computing $f$ on the layered graph, i.e.  $f_{\mathcal{G}^L}$.
\begin{figure}[h]
  \begin{center}
  \subfigure[]{
    \begin{tikzpicture}[->, >=stealth', auto, semithick, node distance=1.5cm]
\tikzstyle{every state}=[fill=white,draw=black,thick,text=black,scale=.55]
\node[state, fill=gray!30] (0) {$v$};
\node[state]    (1)[right of=0, xshift=.3cm]   {$u$};
\node[state]    (2)[below of=0]   {$w$};
\coordinate (Middle) at ($(0)!0.5!(1)$);	
\node (t1)[below of = 0, xshift=.6cm, yshift=0.3cm] {$t=1$};

\node[state, fill=gray!30]    (10)[right of=1] {$v$};
\node[state, fill=gray!30]    (11)[right of=10, xshift=.3cm]   {$u$};
\node[state]    (12)[below of=10]   {$w$};
\node (t2)[below of = 10, xshift=.6cm, yshift=0.3cm] {$t=2$};

\node[state, fill=gray!30]    (20)[below of=2, yshift=-0.3cm] {$v$};
\node[state, fill=gray!30]    (21)[right of=20, xshift=.3cm] {$u$};
\node[state, fill=gray!30]    (22)[below of=20]   {$w$};
\node (t3)[below of = 20, xshift=.6cm, yshift=0.3cm] {$t=3$};

\draw[->] (0) -- node{$p_{vu}$}  (1);
\draw[->] (1) -- node{$p_{uv}$} (2);

\draw[->] (10) --  node{$p_{vu}$} (11);
\draw[->] (11) --  node{$p_{uv}$} (12);

\draw[->] (20) --  node{$p_{vu}$} (21);
\draw[->] (21) --  node{$p_{uw}$} (22);

\end{tikzpicture}
  } 
  \subfigure[]{
    \begin{tikzpicture}[->, >=stealth', auto, semithick, node distance=2cm]
\tikzstyle{every state}=[fill=white,draw=black,thick,text=black,scale=0.73]
\node[state, fill=gray!30]    (0)   {$v_1$};
\node[state]    (1)[below of=0]   {$u_1$};
\node[state]    (2)[below of=1]   {$w_1$};
\node[state, fill=gray!30]    (3)[right of=0]   {$v_2$};
\node[state, fill=gray!30]    (4)[below of=3]   {$u_2$};
\node[state]    (5)[below of=4]   {$w_2$};
\node[state, fill=gray!30]    (6)[right of=3]   {$v_3$};
\node[state, fill=gray!30]    (7)[below of=6]   {$u_3$};
\node[state, fill=gray!30]    (8)[below of=7]   {$w_3$};

\path
(0) edge[left,right]    node{$p_{vu}$}      (4)
(3) edge[left,right]    node{$p_{vu}$}      (7)
(1) edge[left,right]    node{$p_{uw}$}      (5)
(4) edge[left,right]    node{$p_{uw}$}      (8)
(0) edge[left,above, dashed]    node{$1$}      (3)
(1) edge[left,above, dashed]    node{$1$}      (4)
(2) edge[left,above, dashed]    node{$1$}      (5)
(3) edge[left,above, dashed]    node{$1$}      (6)
(4) edge[left,above, dashed]    node{$1$}      (7)
(5) edge[left,above, dashed]    node{$1$}      (8);
\end{tikzpicture}
  }
  \vspace{-1.3em}
  \caption{Influence propagation representation over (a) Original graph $\mathcal{G}$ and (b) Layered graph $\mathcal{G}^L$. The shaded nodes illustrate the active nodes in the graph. $p_{vu}$ represents the propagation probability between $v$ and $u$. In both cases, the nodes can only switch from being inactive to being active.}
  \label{fig:graphs}
\end{center}
\vspace{-.5em}
\end{figure}
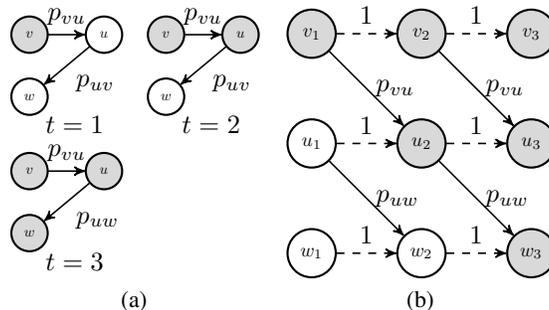 
\begin{lemma}
  \label{lem:1}
  For seed set $\CS$ (with time indices) and realization $\phi$, it holds that $\tilde{f}_{\mathcal{G}}(\CS,\phi) = f_{\mathcal{G}^L}(\CS,\phi)$.
\end{lemma}
\begin{proof}
  It suffices to remark that the number of active nodes on layer $\CL_t$ is equal to the number of active nodes on $\mathcal{G}$ at time $t$.
  Summing up the active nodes of each layer $\CL_i$ is the same by applying $f$ on $\mathcal{G}^L$, which is equivalent to summing up the number of active nodes on $\mathcal{G}$ at each time-step.
\end{proof}
In our model, the time dependency is even stronger compared to previous models.
\emph{Partial realizations} $\psi$ should now indicate the status of observed nodes and edges as well as the corresponding timesteps, as nodes can be active over multiple timesteps and edges can be crossed multiple times.
Actually, we need to know up to which time step the $\psi$ contains observations.
This leads to the next definition.
\begin{definition}
Let $\Psi$ be the set of all possible partial realizations. Time function $\mathcal{T} : \Psi \rightarrow \{1,\dots,T\}$ returns, for a particular $\psi$, the largest time
index from observed nodes and edges, and $1$ if $\psi = \emptyset$.
\end{definition}
In a nutshell, choosing $v$ as a seed node having observed $\psi$ with $\mathcal{T}(\psi) = t \leq T$, is the same as choosing $v_t$ as a seed node in the layered graph, since the process
is now at time $t$. In this point, let us provide a last definition.
\begin{definition}
  \label{def:5}
  The marginal gain of choosing $v$ as a seed node, having observed $\psi$ with $\mathcal{T}(\psi) = t$, and for the ground truth realization $\phi$ of the network, is defined as:
  \begin{equation*}
    \delta_{\phi}(v|\psi) \defn \tilde{f}_{\mathcal{G}}(dom(\psi)\cup \{v_t\},\phi) - \tilde{f}_{\mathcal{G}}(dom(\psi),\phi).
  \end{equation*}
\end{definition}
The aforementioned definition is useful for the analysis of the next three lemmas. Lemma~\ref{lem:2} is a \textit{markovian result on layers}. It shows
that, to evaluate $\delta_{\phi}(v|\psi)$, we only need information from the current layer, $\CL_{\mathcal{T}(\psi)}$. Information
from previous layers, $\CL_1, \dots, \CL_{\mathcal{T}(\psi) - 1}$, have no impact on the marginal gain of adding $v$ to seed nodes at time $\mathcal{T}(\psi)$. 
On the other hand, Lemmas~\ref{lem:3} and \ref{lem:4} are inequalities over  $\delta_{\phi}(\cdot|\psi)$, that will be central in the proofs of Section~\ref{sec:TheroreticalAnalysis}.
\begin{lemma}
  \label{lem:2}
  The marginal gain of choosing $v$ as a seed node on $\mathcal{G}^L$, under  partial realization $\psi$ with $\mathcal{T}(\psi) = t$, is given by:
  $\delta_{\phi}(v|\psi) = f_{\mathcal{G}^L}(\left[ \CL_t \cap dom(\psi) \right]\cup \{v_t\},\phi) - f_{\mathcal{G}^L}( \CL_t \cap dom(\psi),\phi).$
\end{lemma}
\begin{proof}
  Based on Def.~\ref{def:5} and Lem.~\ref{lem:1}, it holds that:
  $$\delta_{\phi}(v|\psi) = f_{\mathcal{G}^L}(dom(\psi)\cup \{v_t\},\phi) - f_{\mathcal{G}^L}(dom(\psi),\phi).$$
  Given a set $\CS$ of seed nodes on the $\mathcal{G}^L$, the utility function $f$ is given by $$f_{\mathcal{G}^L}(\CS,\phi) = \sum_{t'=1}^T |\sigma(\CS,\phi) \cap \CL_{t'} |.$$
  Then we get that:
\begin{align*}
  \delta_{\phi}(v|\psi) &= \sum_{t'=1}^T |\sigma(dom(\psi)\cup \{v_t\},\phi) \cap \CL_{t'} | - \sum_{t'=1}^T |\sigma(dom(\psi),\phi) \cap \CL_{t'} | \\
                        &= \sum_{t'=\textcolor{red}{t}}^T |\sigma(dom(\psi)\cup \{v_t\},\phi) \cap \CL_{t'} | - \sum_{t'=\textcolor{red}{t}}^T |\sigma(dom(\psi),\phi) \cap \CL_{t'} | \\
                        &= f_{\mathcal{G}^L}([ \CL_t \cap dom(\psi)]\cup \{v_t\},\phi) - f_{\mathcal{G}^L}( \CL_t \cap dom(\psi),\phi).
\end{align*}
The second equality holds due to the fact that the network $\mathcal{G}^L$ is feedforward, which means that the node $v_t$ can only influence nodes on the subsequent layers: $\CL_{t+1}, \dots,  \CL_{T}$.
\end{proof}
\begin{lemma} \label{lem:3}
  For partial realizations $\psi \subseteq \psi'$ with $\mathcal{T}(\psi) = \mathcal{T}(\psi') = t$ and any $v \in V$, we get $\delta_{\phi}(v|\psi) \geq \delta_{\phi}(v|\psi').$
\end{lemma}
\begin{proof}
  Let $\CR(v_t,\phi)$ denotes the set of nodes that can be reached from node $v_t$ via a path consisting of live edges, under realization $\phi$. For any $\CA \subseteq \CL_t$ (layer $t$
  of $\mathcal{G}^L$), we have $f_{\mathcal{G}^L}(\CA,\phi) = | \cup_{v \in \CA} \CR(v,\phi) |$.
  Let us now consider the quantity  $f_{\mathcal{G}^L}( \CA \cup \{v_t\},\phi) - f_{\mathcal{G}^L}(\CA,\phi)$ to be equal to the number of elements of $\CR(v_t,\phi)$ that are not already contained in $\cup_{v \in \CA} \CR(v,\phi)$.
  Clearly, this quantity is larger or equal to the number of elements of $\CR(v_t,\phi)$ that are not contained in the \textit{bigger} set $\cup_{v \in \CB} \CR(v,\phi)$, for any $\CA \subseteq \CB \subseteq \CL_t$. Therefore, it holds that:
  \begin{equation*}
    f_{\mathcal{G}^L}(\CA\cup \{v_t\},\phi) - f_{\mathcal{G}^L}(\CA,\phi) \geq f_{\mathcal{G}^L}( \CB \cup \{v_t\},\phi) - f_{\mathcal{G}^L}(\CB,\phi).
  \end{equation*}
  Setting  $\CA = \CL_t \cap dom(\psi)$, $\CB = \CL_t \cap dom(\psi')$ and using Lem.~\ref{lem:2}, we get: $\delta_{\phi}(v|\psi) \geq \delta_{\phi}(v|\psi')$.
\end{proof}
\begin{lemma} \label{lem:4}
  For partial realizations $\psi \subseteq \psi'$ with $\mathcal{T}(\psi) < \mathcal{T}(\psi')$ and any $v \in \CV \setminus dom(\psi')$, we get $\delta_{\phi}(v|\psi) \geq 1 + \delta_{\phi}(v|\psi').$
\end{lemma}
\begin{proof}
Let us first consider w.l.o.g. that  $\mathcal{T}(\psi) = t$ and $\mathcal{T}(\psi') = t + 1$. Here, the node $v_t$ is activated in $\mathcal{G}^L$, after observing $\psi$. Since $v \notin \text{dom}(\psi')$ by assumption, then $v \notin \text{dom}(\psi)$ and therefore $v$ is not already active. Let $\psi_+$ denote the partial realization combining $\psi$ and observing one more step of the process - from layer $t$
to layer $t+1$ - without adding any seed node, w.r.t. $\phi$. Also, let $A$ denote the set of active nodes of layer $t+1$ \textit{that would not have been activated if $u_t$ has not been selected as seed node}, except $v_{t+1}$. In this scenario, we get: 
\begin{equation*}
  \delta_{\phi}(v|\psi) = 1 + \delta_{\phi}(v \cup A |\psi_+) \geq 1 + \delta_{\phi}(v|\psi_+) \geq 1 + \delta_{\phi}(v|\psi').
\end{equation*}
The first equality comes from the fact that $\mathcal{G}^L$ is feedforward, therefore activating $v$ brings a reward of 1 at time $t$, plus the reward from the future.
The second inequality is due to the monotonicity of the set function $\delta_{\phi}(.|\psi_+)$.

The last inequality  holds due to the fact that $\psi_+ \subseteq \psi'$ (application of Lem.~\ref{lem:3}). Indeed, since $\psi \subseteq \psi'$, all nodes observed to be active by $\psi$ at time $t$ are also observed to be active by $\psi'$. Therefore, if the status of an edge from layer $t$ to $t+1$ is observed under $\psi_+$, it is also observed under $\psi'$.
As a consequence, we notice that $\text{dom}(\psi_+) \cap \mathcal{L}_{t+1} \subseteq \text{dom}(\psi') \cap \mathcal{L}_{t+1}$, i.e., all the nodes observed to be active by $\psi_+$ on layer $t+1$ of $\mathcal{G}^L$ are also observed to be active by $\psi'$.
In this point, it should be recalled that $\psi$, $\psi_+$ and $\psi'$ are all consistent w.r.t. the same ground truth realization $\phi$. 

Finally, it can be verified that this inequality still holds for $\mathcal{T}(\psi') = t + x$ with $x>1$. Actually, tighter inequalities could be obtained for $x>1$, but the inequality of this Lemma is more simple, and sufficient for the proof of Theorem~\ref{Theorem1}.
\end{proof}

\section{Theoretical Guarantees for the Myopic Adaptive Greedy Strategy}
\label{sec:TheroreticalAnalysis}

In this section, we introduce the  \textit{myopic adaptive greedy policy}. Using our layered graph representation, we prove that this policy is guaranteed to provide a $(1 - 1/e)$-approximation of the optimal policy, in the framework presented in Sec.~\ref{sec:MyopicFeedback}.  

\noindent\textbf{Myopic adaptive greedy policy} The \textit{myopic adaptive greedy policy} starts with an empty set $\CS = \emptyset$, and repeatedly chooses as seed the node that gives the maximum expected marginal gain under partial realization $\psi$.
If the graph is too large, expected marginal gains can be estimated via Monte Carlo simulations as in \citet{kdd/KempeKT03}.
For simplicity reasons, we assume w.l.o.g. that only one seed node is selected at each time step.
A sketch of our policy is presented in Alg.~\ref{alg:1}.

\begin{algorithm}[h!]
  \caption{Myopic adaptive greedy policy}
  \label{alg:1}
  \begin{algorithmic}[1]
    \REQUIRE $\mathcal{G}, T$
    \STATE $\psi \leftarrow \emptyset,~ \CS \leftarrow \emptyset$
    \FOR{t = 1 {\bfseries to} T}
    \STATE Compute $\Delta_{\tilde{f}}(v|\psi), \forall v \in \CV \setminus \CS $
    \STATE Select $v^* \in \argmax\limits_{v \in \CV \setminus S}  \Delta_{\tilde{f}}(v|\psi)$
    \STATE $\CS \leftarrow \CS  \cup \{v^*\}$
    \STATE Update $\psi$ observing (one-step) myopic feedback
    \STATE $\CS \leftarrow \CS  \cup dom(\psi)$
    \ENDFOR
    \RETURN $\CS$  (final set of influenced nodes)
  \end{algorithmic}
\end{algorithm}

\subsection{Theoretical guarantees}

We are now ready to formally state our main result that constitutes an approximation guarantee for the proposed strategy.
Actually, the key point of our proof is to check that the proposed utility function $\tilde{f}_{\CG}$ is adaptive monotone and adaptive submodular w.r.t. $p(\phi)$.
These properties in conjunction with the result of \citet{Golovin:2011} complete our proof.
\begin{theorem}
  \label{Theorem1}
  The adaptive greedy policy $\pi^{\text{g}}$ obtains at least $(1 - 1/e)$ of the value of the best policy for the adaptive influence maximization problem under the \emph{modified} IC
  model with myopic feedback and $\tilde{f}$ as utility function. In other words, if $\tilde{f}_{avg}(\pi^{\text{g}})  \defn \mathbb{E}_{\Phi} [\tilde{f}_{\mathcal{G}}(E(\pi,\Phi),\Phi)]$, we get that:
  \begin{equation*}
    \tilde{f}_{avg}(\pi^{\text{g}})  \geq (1 - 1/e) \tilde{f}_{avg}(\pi^*).
  \end{equation*}
\end{theorem}
\begin{proof}
  Our objective is to prove that the utility function $\tilde{f}_{\mathcal{G}}$ is adaptive monotonic and adaptive submodular w.r.t. $p(\phi)$.
  Adaptive monotonicity is straightforward, since $\tilde{f}_{\mathcal{G}}(\cdot,\phi)$ is itself monotonic $\forall \phi$.

  Let us consider two subrealizations $\psi$ and $\psi'$ with $\psi \subseteq \psi'$ and $u \notin \text{dom}(\psi')$. To prove that the proposed utility function $\tilde{f}_{\mathcal{G}}$ is adaptive submodular, we need to verify that $\Delta(u|\psi) \geq \Delta(u|\psi')$, i.e., 
    \begin{align*}
      &\mathbb{E}_{\Phi}\left[ \tilde{f}_{\mathcal{G}}(\text{dom}(\psi) \cup \{u_{\mathcal{T}(\psi)}\},\Phi) - \tilde{f}_{\mathcal{G}}(\text{dom}(\psi),\Phi) | \Phi \sim \psi \right] \geq \\
      &\mathbb{E}_{\Phi}\left[ \tilde{f}_{\mathcal{G}}(\text{dom}(\psi') \cup \{u_{\mathcal{T}(\psi')}\},\Phi) - \tilde{f}_{\mathcal{G}}(\text{dom}(\psi'),\Phi) | \Phi \sim \psi' \right].
    \end{align*}
  According to  Def.~\ref{def:5}, we need to check that:
  \begin{equation*}
    \sum_{\phi} p(\phi | \psi) \delta_{\phi}(u|\psi) \geq \sum_{\phi} p(\phi | \psi') \delta_{\phi}(u|\psi'),
  \end{equation*}
  where $p(\phi|\psi) \defn \mathbb{P}[\Phi = \phi| \Phi \sim \psi]$. Note that $p(\phi|\psi) = 0$ if $\phi$ is inconsistent with $\psi$. Otherwise, if $\phi \sim \psi$, we have:
    \begin{equation*}
      p(\phi|\psi) = \prod_{t = 1}^{T-1} \prod \limits_{\underset{\text{unobserved by } \psi}{(v_{t},w_{t+1}) \in \mathcal{E}_{\mathcal{G}^L}}} p_{vw}^{X_{v_{t}w_{t+1}}} (1 - p_{vw})^{1 - X_{v_{t}w_{t+1}}},
    \end{equation*}
  where $\mathcal{E}_{\mathcal{G}^L}$ is the set of edges of $\mathcal{G}^L$ (the layered graph representation of $\mathcal{G}$), and $X_{v_{t}w_{t+1}} \sim \mathcal{B}(p_{vw})$ is a Bernoulli r.v. whose realization indicates whether the edge $(v_{t},w_{t+1})$ of $\mathcal{G}^L$ is live or dead in the ground truth realization $\phi$. More specifically, it indicates if active node $v_t$ succeeds to activate its neighbor $w$ at time $t+1$, or not.

  In order to obtain our result, let us first recall that $\delta_{\phi}(u|\psi) \geq \delta_{\phi}(u|\psi')$. There are three possible different situations, depending on $\mathcal{T}(\psi)$ and $\mathcal{T}(\psi')$. The first scenario, $\mathcal{T}(\psi) > \mathcal{T}(\psi')$, is actually impossible, since it will violate our assumption that $\psi \subseteq \psi'$. For the second where $\psi \subseteq \psi'$ with $\mathcal{T}(\psi) = \mathcal{T}(\psi')$ a direct application of Lemma~\ref{lem:3} gives that $\delta_{\phi}(u|\psi) \geq \delta_{\phi}(u|\psi')$. In the last case, $\psi \subseteq \psi'$ with $\mathcal{T}(\psi) < \mathcal{T}(\psi')$, we get $\delta_{\phi}(u|\psi) \geq 1 +\delta_{\phi}(u|\psi')$ according to Lemma~\ref{lem:4}. 

  \noindent\textbf{Proof of $\Delta(u|\psi) \geq \Delta(u|\psi')$ when $\mathcal{T}(\psi) = \mathcal{T}(\psi')$: }
  Using the aforemetioned results, we will prove that $\Delta(u|\psi) \geq \Delta(u|\psi')$ in the scenario where $\mathcal{T}(\psi) = \mathcal{T}(\psi')$. It can be easily verified that if $\mathcal{T}(\psi) = T$, the equality $\Delta(u|\psi) = \Delta(u|\psi') = 1$ holds. Now, we focus on $\mathcal{T}(\psi) < T$.

  To begin, let us introduce some new objects. Let $\mathcal{\tilde{G}}^L$ be a truncated version of $\mathcal{G}^L$ where we removed the layers and edges before time step $\mathcal{L}_{\mathcal{T}(\psi)}$. Equivalently, $\mathcal{\tilde{G}}$ is a graph with the same structure as $\CG$, but we start the IM problem at $t = \CT(\psi)$ ($=\CT(\psi')$) instead of $t=1$ while some of the nodes are already active at the beginning of the process (the ones observed to be active on $\CL_{\CT(\psi)}$). Finally, let $\tilde{\phi}$ be the truncated version of $\phi$ on $\mathcal{\tilde{G}}^L$, i.e. all Bernoulli r.v. on the edges between layers $\CL_{\CT(\psi)}$ and $\CL_T$ have the same status. We denote as  $\phi \sim \tilde{\phi}$, the consistency between $\phi$ and $\tilde{\phi}$. We also have:
  \begin{equation*}
    \tilde{p}(\tilde{\phi}) = \prod_{t = \mathcal{T}(\psi)}^{T-1} \prod \limits_{(v_{t},w_{t+1}) \in \mathcal{E}_{\mathcal{\tilde{G}}^L}} p_{vw}^{X_{v_{t}w_{t+1}}} (1 - p_{vw})^{1 - X_{v_{t}w_{t+1}}}.
  \end{equation*}
  Now let us go back to our primary goal where we have:
  \begin{equation*}
    \Delta(u|\psi) - \Delta(u|\psi') = \sum_{\phi} p(\phi | \psi) \delta_{\phi}(u|\psi) - \sum_{\phi} p(\phi | \psi') \delta_{\phi}(u|\psi').
  \end{equation*}
  The probabilities $p(\phi | \psi)$ and $p(\phi | \psi')$ are defined for the realizations $\phi \sim \psi$ and $\phi \sim \psi'$, respectively.
  However, according to Lemma~\ref{lem:2}, randomness on marginal gains comes only from the unknown statuses of the edges from layers $\mathcal{L}_{\mathcal{T}(\psi)}$ to $\mathcal{L}_{T}$ of the layered graph representation $\mathcal{G}^L$. The actual statuses (live or dead) of edges connecting past layers do not have any impact at $\delta_{\phi}(u|\psi)$ and $\delta_{\phi}(u|\psi')$, respectively.
  Since $\tilde{p}(\tilde{\phi}) = \sum\limits_{\phi \sim \phi'} p(\phi | \psi)$, and
  \begin{align*}
    \delta_{\tilde{\phi}}(u|\psi) &= \tilde{f}_{\mathcal{\tilde{G}}}(\text{dom}(\psi) \cup \{u_{\mathcal{T}(\psi)}\},\tilde{\phi}) - \tilde{f}_{\mathcal{\tilde{G}}}(\text{dom}(\psi),\tilde{\phi}) \\
                                  &= \tilde{f}_{\mathcal{G}}(\text{dom}(\psi) \cup \{u_{\mathcal{T}(\psi)}\},\phi) - \tilde{f}_{\mathcal{G}}(\text{dom}(\psi),\phi) \\
                                  &= \delta_{\phi}(u|\psi),
  \end{align*}
  we conclude that:
   $ \sum_{\phi} p(\phi | \psi) \delta_{\phi}(u|\psi) = \sum_{\tilde{\phi}} \tilde{p}(\tilde{\phi}) \delta_{\tilde{\phi}}(u|\psi).$
  In the same way, we get that $\sum_{\phi} p(\phi | \psi') \delta_{\phi}(u|\psi') = \sum_{\tilde{\phi}} \tilde{p}(\tilde{\phi}) \delta_{\tilde{\phi}}(u|\psi')$, with $\delta_{\phi}(u|\psi') = \delta_{\tilde{\phi}}(u|\psi')$.
  Therefore, we derive that:
    \begin{align*}
      \Delta(u|\psi) - \Delta(u|\psi') &= \sum_{\phi} p(\phi | \psi) \delta_{\phi}(u|\psi) - \sum_{\phi} p(\phi | \psi') \delta_{\phi}(u|\psi') \\
                                       &= \sum_{\tilde{\phi}} \tilde{p}(\tilde{\phi}) \delta_{\tilde{\phi}}(u|\psi) - \sum_{\tilde{\phi}} \tilde{p}(\tilde{\phi}) \delta_{\tilde{\phi}}(u|\psi') \\
                                       &= \sum_{\tilde{\phi}} \tilde{p}(\tilde{\phi}) \left(\delta_{\tilde{\phi}}(u|\psi) -  \delta_{\tilde{\phi}}(u|\psi') \right) \geq 0.
    \end{align*}
  The last inequality holds, as $\delta_{\tilde{\phi}}(u|\psi) \geq  \delta_{\tilde{\phi}}(u|\psi')$.
  
  \noindent\textbf{Proof of $\Delta(u|\psi) \geq \Delta(u|\psi')$ when $\mathcal{T}(\psi') = \mathcal{T}(\psi) + 1$}
  Let us now focus on the scenario where $\mathcal{T}(\psi) < \mathcal{T}(\psi')$. Initially, we consider the case where $\mathcal{T}(\psi') =\mathcal{T}(\psi) + 1$. We define $\tilde{\phi}$ and $\tilde{p}(\tilde{\phi})$ as before but w.r.t. $\psi'$ (i.e. the first layer of $\tilde{\mathcal{G}}^L$ is $\mathcal{L}_{\mathcal{T}(\psi')}$). It is important to remark that:
    \begin{align*}
  \sum_{\phi \sim \tilde{\phi}} p(\phi | \psi') = \sum_{\phi \sim \tilde{\phi}} \prod_{t = 1}^{T-1} \prod \limits_{\underset{\text{unobserved by } \psi'}{(v_{t},w_{t+1}) \in \mathcal{E}_{\mathcal{G}^L}}} p_{vw}^{X_{v_{t}w_{t+1}}} (1 - p_{vw})^{1 - X_{v_{t}w_{t+1}}}& \\ 
                                                =\sum_{\phi \sim \tilde{\phi}} \left(\prod_{t = 1}^{ \mathcal{T}(\psi')-1} \prod \limits_{\underset{\text{unobserved by } \psi'}{(v_{t}w_{t+1}) \in \mathcal{E}_{\mathcal{G}^L}}} p_{vw}^{X_{v_{t}w_{t+1}}} (1 - p_{vw})^{1 - X_{v_{t}w_{t+1}}} \right)& \\
                                                \left( \underbrace{\prod_{t = \mathcal{T}(\psi')}^{T-1} \prod \limits_{(v_{t},w_{t+1}) \in \mathcal{E}_{\mathcal{\tilde{G}}^L}} p_{vw}^{X_{v_{t}w_{t+1}}} (1 - p_{vw})^{1 - X_{v_{t}w_{t+1}}}}_{ = \tilde{p}(\tilde{\phi}) \text{ (same for all $\phi \sim \tilde{\phi}$)}} \right)& \\
                                                = \tilde{p}(\tilde{\phi}) \underbrace{\sum_{\phi \sim \tilde{\phi}}\prod_{t = 1}^{ \mathcal{T}(\psi')-1} \prod \limits_{\underset{\text{unobserved by } \psi'}{(v_{t}w_{t+1}) \in \mathcal{E}_{\mathcal{G}^L}}} p_{vw}^{X_{v_{t}w_{t+1}}} (1 - p_{vw})^{1 - X_{v_{t}w_{t+1}}}}_{=1} = \tilde{p}(\tilde{\phi}).&
    \end{align*}
  In a similar way, we get that $\sum\limits_{\phi \sim \tilde{\phi}} p(\phi | \psi) = \tilde{p}(\tilde{\phi})$.
  Thus, we get that
    \begin{align*}
      \Delta(u|\psi) &= \sum_{\phi} p(\phi | \psi) \delta_{\phi}(u|\psi) 
                     \geq \sum_{\phi} p(\phi | \psi) \Big(1 + \delta_{\phi}(u|\psi') \Big) \\
                     &= \underbrace{\sum_{\phi} p(\phi | \psi)}_{= 1} + \sum_{\phi} p(\phi | \psi) \delta_{\phi}(u|\psi') 
                     = 1 + \sum_{\phi} p(\phi | \psi) \delta_{\phi}(u|\psi') \\
                     &= 1 + \sum_{\tilde{\phi}} \underbrace{\sum_{\phi \sim \tilde{\phi}} p(\phi | \psi)}_{= \tilde{p}(\tilde{\phi)}} \underbrace{\delta_{\phi}(u|\psi')}_{= \delta_{\tilde{\phi}}(u|\psi')} 
                     = 1 + \sum_{\tilde{\phi}} \tilde{p}(\tilde{\phi}) \delta_{\tilde{\phi}}(u|\psi'),
    \end{align*}
  and
    \begin{align*}
      \Delta(u|\psi') &= \sum_{\phi} p(\phi | \psi') \delta_{\phi}(u|\psi') = \sum_{\tilde{\phi}} \underbrace{\sum_{\phi \sim \tilde{\phi}} p(\phi | \psi')}_{= \tilde{p}(\tilde{\phi)}} \underbrace{\delta_{\phi}(u|\psi')}_{= \delta_{\tilde{\phi}}(u|\psi')} 
                      = \sum_{\tilde{\phi}} \tilde{p}(\tilde{\phi}) \delta_{\tilde{\phi}}(u|\psi').
    \end{align*}
  Therefore, we conclude that: $\Delta(u|\psi) - \Delta(u|\psi') \geq 0$.

  \noindent\textbf{Proof of $\Delta(u|\psi) \geq \Delta(u|\psi')$ when $\mathcal{T}(\psi') = \mathcal{T}(\psi) + x$ with $x>1$}

  So far, we focused on the case $\mathcal{T}(\psi') = \mathcal{T}(\psi) + 1$. Actually, it is quite straightforward to extend results to the scenario where we consider partial realizations $\psi_{t}$, $\psi_{t+x}$ with $\psi_{t} \subseteq \psi_{t+x}$, $\mathcal{T}(\psi_t) = t$ and $\mathcal{T}(\psi_{t+x}) = t + x$ with $x >1$.

  Let $\psi_{t+1}$, $\psi_{t+2}$, ..., $\psi_{t+x-1}$ denote partial realizations such that $\psi_{t} \subseteq \psi_{t+1} \subseteq \psi_{t+2} ... \subseteq \psi_{t+x-1} \subseteq \psi_{t+x}$. Using telescoping sum and our previous result, we obtain that:
  \begin{equation*}
    \Delta(u|\psi_t) - \Delta(u|\psi_{t+x}) = \sum_{i=0}^{x-1} \Big(\underbrace{\Delta(u|\psi_{t+i}) - \Delta(u|\psi_{t+i+1})}_{\geq 0} \Big),
  \end{equation*}
  that concludes our proof.
\end{proof}

This is the first time that such inequality is demonstrated on the adaptive setting under \emph{myopic feedback}. 
Using the generalization of the result of \citep{Golovin:2011}, we also retrieve the $(1 - e^{-\ell/\alpha k})$ bound for any $\alpha$-approximate ($\ell$-truncated) greedy policies.
It can be also verified that the bound of Theorem~\ref{Theorem1} is still valid even if we select more than one seed node at each time step.

\subsection{Modified IC model hypotheses}

In this point we discuss the two central hypotheses of the proposed modified IC model:  an active node i) has multiple opportunities to influence its neighbors, and ii) cannot be randomly deactivated over time. Actually, we demonstrate that the proposed utility function $\tilde{f}$ is adaptive submodular only in the case where these two assumptions hold. \\

\noindent\textbf{Utility function $\tilde{f}$ under standard IC model}
Let us now consider the \emph{standard} IC model with myopic feedback and $\tilde{f}$ as utility function. Actually, removing the assumption that active nodes have multiple opportunities to influence its neighbors, we get the \emph{standard} IC model where each active node has a \emph{unique} chance to influence its neighbors.

\begin{lemma} \label{lem:5}
  The utility function $\tilde{f}$ is not adaptive submodular under the \emph{standard} IC model with myopic feedback.
\end{lemma}
\begin{proof}
  Let us consider the network shown in Fig.~\ref{fig:lem5graph} that consists of two nodes $u$ and $v$, with $p_{uv} \defn p \in [0,1]$.  We assume that $T = 3$ and that node $u$ is already active at $t=1$. As we consider the \emph{standard}  IC model, $u$ has an \textit{unique chance} to influence $v$, at $t=2$, succeeding with probability $p$. Let also $\psi = \emptyset$: we have no information on the unique edge of this graph, we only know that $u$ is active at $t=1$. Therefore, we have $\mathcal{T}(\psi) = 1$. Moreover, let $\psi'$ contains the information that $u$ is active at $t = 1$ and that it has failed to influence $v$ at $t=2$ ($\mathcal{T}(\psi') = 2$). Since $u$ has a unique chance to influence $v$, there is no more randomness about the ground truth realization $\phi$ at this point. We have $\psi \subseteq \psi'$.

  Considering node $v$ as a seed node given subrealization $\psi$, we get that:
  \begin{equation*}
    \Delta(v|\psi) = p \times 1 + (1-p) \times 3.
  \end{equation*}
  Indeed, if the edge $(u,v)$ is dead (probability $1-p$), the marginal gain of activating $v$ at $t=1$ is equal to $3$ (nodes $v_1$, $v_2$ and $v_3$ will be activated in the layered graph). On the other hand, if the edge $(u,v)$ is live (probability $p$), $v$ will have been actived at time steps $t=2$ and $t=3$ even without the activation of $v$ at $t=1$. Therefore, the only marginal gain comes from the activation of $v$ at $t=1$, that is equal to $1$.
  Similarly, we get that
  $  \Delta(v|\psi') = 2.$
  Choosing $v$ as a seed note after observing $\psi'$, i.e., at $t=2$, leads to a marginal gain equal to $2$, rewarding the activation of $v$ at $t=2$ and $3$.

  It can be easily verified that $\Delta(v|\psi) \geq \Delta(v|\psi')$ iff $p \leq 0.5$. Therefore, the adaptive submodularity property holds only in the case where $p > 0.5$.
 \begin{figure}[t]
   \centering
   \begin{tikzpicture}[->, >=stealth', auto, semithick, node distance=2.5cm]
\tikzstyle{every state}=[fill=white,draw=black,thick,text=black,scale=1.]
\node[state] (0) {$u$};
\node[state] (1)[right of=0] {$v$};
\path
(0) edge[left,above] node{$p$} (1);
\end{tikzpicture}
   \caption{Toy graph used as counterexample at Lem.~\ref{lem:5}.}
   \label{fig:lem5graph}
 \end{figure}
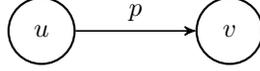
\end{proof}

It should be also mentioned that the utility function $\tilde{f}$ is adaptive submodular in the aforementioned network if we consider the \emph{modified} IC diffusion model. More specifically, we get that
\begin{align*}
  \Delta(v|\psi) &= p^2 + p(1-p) + 2 (1-p)p + 3(1-p)^2 \\
                 &= p^2 + 3p(1-p) + 3(1-p)^2
\end{align*}
and $\Delta(v|\psi') = p + 2(1-p) = 2-p.$ Therefore, we can easily check that the inequality $\Delta(v|\psi) \geq \Delta(v|\psi')$ holds for any $p \in [0,1]$.

\noindent\textbf{Non-Progressive Adaptive Submodular IM}
In this point, we examine the scenario where the second main hypothesis of the model (active nodes can not be deactivated randomly) does not hold anymore. Actually, the application itself can determine if this hypothesis is realistic or not. In the case of our layered graph representation, we can easily relax this assumption, by replacing the ``1'' with a random probability over the edges between the same nodes. Our model along with the main notations are still well defined under this relaxation.

However, it appears that it destroys the reasoning of the proof of our main result (Theorem~\ref{Theorem1}), as the utility function $\tilde{f}_{\mathcal{G}}$ is no longer adaptive submodular.
Additionally, we show that the adaptive submodularity property is also violated even in the case of the \emph{full-adoption feedback} by using the standard IM utility function $f_{\mathcal{G}}$.
\begin{lemma}  \label{lem:6}
  Forcing active nodes to remain active throughout the process constitutes a \textit{necessary condition} to verify the adaptive submodularity property of:\\
i) $\tilde{f}_{\mathcal{G}}$ in the modified IC model with myopic feedback;\\
ii) $f_{\mathcal{G}}$ in the standard IC model with full-adoption feedback.
\end{lemma}
\begin{proof}
  i) In the case of the modified myopic feedback model, we consider the \textit{layered graph} of Fig.~\ref{fig:GraphExamples}(a) that consists of six random edges.
  There are $2^6 = 64$ ground truth graphs, each of them being obtained with probability $1/64$ since edges are independent Bernoulli r.v., $\mathcal{B}(1/2)$.
  We want to add $v$ to the set of seed nodes. 
  Now, consider $\psi$ where we only know that $u$ is activated at $t = 1$ ($\mathcal{T}(\psi) = 1$), and $\psi'$ where we also observed that 
  $(u_1, u_2)$ and $(u_1,v_2)$ are dead edges ($\mathcal{T}(\psi') = 2$). Clearly, $\psi \subseteq \psi'$.
  A simple decomposition of all possible ground truth graphs leads to $\Delta(v|\psi) = 1 + [2\times\frac{6}{64} + 1\times \frac{10}{64} + 0 \times \frac{48}{64}] = \frac{86}{64}$:
  a reward of $1$ for activating $v_1$ and possibly a marginal gain of adding $v_2$ and $v_3$ ($0$ in $48$ ground truth realizations, $1$ in $10$ of them,  $2$ in $6$ of them).
  We also get that $\Delta(v|\psi') = 1 + \frac{1}{2} = 1.5$: $v_2$ is active (seed) while $v_3$ is active with probability $1/2$.
  Therefore, $\Delta(v|\psi') > \Delta(v|\psi)$.

  
  ii) Let us consider the graph of Fig.~\ref{fig:GraphExamples}(b), where active nodes have a probability of $1/2$ to be deactivated at each time.
  Recall that our utility function is now the number of activated nodes at the end of the process (standard IC), and let $T = 2$.
  Suppose also that we want to choose $v$ as seed node under the next two scenarios. At the first one we are at time step $t=1$, so $\psi = \emptyset$. The second scenario assumes that we are at $t=2$ having chosen node $u$ at $t=1$, so $\psi'$ only contains the information that $u$ is activated.
  Thus, we get that $\Delta(v|\psi) < 3$, since $v$, $w$ and $z$ are active at $t=1$, but they have a non-null probability to be deactivated at $t=2$. On the other hand, $\Delta(v|\psi') = 3$ as the process ends right after nodes $v$, $w$ and $z$ are activated via choosing $v$ as seed node. Since $\psi \subseteq \psi'$ and $\Delta(v|\psi') = 3 > \Delta(v|\psi)$, adaptive submodularity is once again violated.
\end{proof}
\begin{figure}[t] 
  \begin{center}
    \subfigure[]{
      \begin{tikzpicture}[->, >=stealth', auto, semithick, node distance=2cm]
\tikzstyle{every state}=[fill=white,draw=black,thick,text=black,scale=.75]
\node[state]    (0)                     {$u_1$};
\node[state]    (1)[below of=0, yshift=0.4cm]   {$v_1$};
\node[state]    (3)[right of=0]                     {$u_2$};
\node[state]    (4)[below of=3, yshift=0.4cm]   {$v_2$};
\node[state]    (6)[right of=3]                     {$u_3$};
\node[state]    (7)[below of=6, yshift=0.4cm]  {$v_3$};

\path
(0) edge[left,right]   node{\tiny{$1/2$}}       (4)
(3) edge[left,right]    node{\tiny{$1/2$}}      (7)
(0) edge[left,above, dashed]    node{\tiny{$1/2$}}      (3)
(1) edge[left,above, dashed]    node{\tiny{$1/2$}}      (4)
(3) edge[left,above, dashed]    node{\tiny{$1/2$}}      (6)
(4) edge[left,above, dashed]    node{\tiny{$1/2$}}      (7);
\end{tikzpicture}
    } \hspace{1em}
    \subfigure[]{
      \begin{tikzpicture}[->, >=stealth', auto, semithick, node distance=2cm]
\tikzstyle{every state}=[fill=white,draw=black,thick,text=black,scale=.75]

\node[state]    (10)[below of=1]                {$v$};
\node[state]    (11)[below of=10, yshift=0.4cm] {$u$};
\node[state]    (13)[right of=10]		{$w$};
\node[state]    (16)[right of=13]               {$z$};

\path
(10) edge[left,above]    node{$1$}      (13)
(13) edge[left,above]    node{$1$}      (16);

\end{tikzpicture}
    }
    \caption{Counterexamples to verify that the spread is not adaptive submodular when the active nodes are allowed to be deactivated.}
    \label{fig:GraphExamples}
  \end{center}
\end{figure}
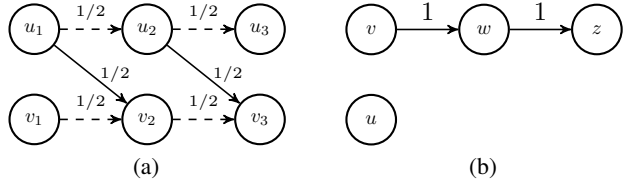

Therefore, the theoretical results presented in our paper and those of \citep{Golovin:2011} are not directly applicable in the case where the active nodes can be deactivated.
However, the hypothesis of active nodes deactivation may be consistent with many applications, including some versions of 
the product placement problem (e.g. customers could reject the product).
In this direction, we are still able to prove a \textit{weaker inequality} at each time step.
We consider the previous framework again, but now active nodes \textit{are allowed to} be deactivated randomly.
At each step $t$, we choose $k_t$ seed nodes from layer $\CL_t$ of $\mathcal{G}^L$, in order to maximize the expected spread in the future having observed which nodes are currently active, 
i.e. active nodes on $\CL_t$. Then, we get the next result.
\begin{lemma} \label{lem:7}
Let $t \in \{1,\dots,T\}$, let $\CS_t \subseteq \CL_t$ the (observed) set of active nodes at time $t$, and consider the following problem: $\CA^* \in \argmax_{\CA \subseteq \CL_t \setminus \CS_t, |A| \leq k} \mathbb{E}_{\Phi}[f_{\CL_t \cup \dots \cup \CL_T}(\CA \cup \CS_t,\Phi)]$. Then, a greedily constructed set $\CA^\text{g} \subseteq \CL_t \setminus \CS_t$ is guaranteed to achieve an $(1-1/e)$-approximation of the optimal set:
$\mathbb{E}_{\Phi}[f_{\CL_t \cup \dots \cup \CL_T}(\CA^\text{g} \cup \CS_t,\Phi)] \geq (1 - 1/e) \mathbb{E}_{\Phi}[f_{\CL_t \cup \dots \cup \CL_T}(\CA^* \cup \CS_t,\Phi)].$
\end{lemma}
\begin{proof} 
  We easily derive from the proof of Lemma~\ref{lem:3} that $f_{\mathcal{G}^L}$ is submodular, for any layered graph $\mathcal{G}^L$ (i.e., also for layered graphs $L_t \cup \dots\cup L_T$).
  Indeed, we proved that for any $\phi$ and $A \subseteq B \subseteq L_t$, 
  $f_{\mathcal{G}^L}(A\cup \{v_t\},\phi) - f_{\mathcal{G}^L}(A,\phi) \geq f_{\mathcal{G}^L}( B\cup \{v_t\},\phi) - f_{\mathcal{G}^L}(B,\phi).$
  Moreover, submodularity being preserved under nonnegative linear combinations, then the objective function of Theorem~\ref{Theorem1} is also itself submodular. Indeed, the expectation is a
  weighted sum of submodular functions, weights being probabilities, according to $p(\phi)$.
  Therefore, we conclude by applying the classical result of \citet{Nemhauser78a}.
\end{proof}

This result is weaker than that of Theorem~\ref{Theorem1}, since it is simply a ``step-by-step'' inequality on each seeding, but not anymore on the entire policy. However, it is free from the assumption that active nodes should remain active.

\section{Empirical Analysis}
\label{sec:Experiments}

We conducted experiments on three social networks from Stanford's SNAP database \citep{snapnets}.
The first one is a small directed \textit{ego network from Twitter} ($|\CV| = 228$, $|\CE| = 9,938$). We also study two medium-size undirected real networks, a \textit{social network from Facebook} 
($|\CV| = 4,039$, $|\CE| = 88,234$) and a \textit{collaboration network from Arxiv General Relativity and Quantum Cosmology} section ($|\CV| = 5,242$, $|\CE| = 28,980$).

Throughout our empirical analysis, we considered the \emph{modified} IC diffusion model with myopic feedback.
Our primary objective is the adaptive selection of $k$ seed nodes, one at each time.
The time horizon is defined as $T = k+1$, i.e. the diffusion process stops one step right after selecting the last seed.
Similar to \citep{kdd/KempeKT03,Gotovos:2015}, we set an identical influence probability at each edge, $p=0.1$. 
All expected marginal gains were estimated via Monte Carlo sampling ($1,000$ simulations). 

\noindent\textbf{Adaptive greedy Vs Heuristic adaptive strategies}
As it is not possible to actually compute the optimal set of influential nodes, we compare the performance of the adaptive greedy strategy w.r.t. three alternative heuristics to identify influential seed nodes.
These heuristics adaptively choose: (i) the node with highest \textit{betweenness centrality}; (ii) the node with \textit{highest degree}; and (iii) a \textit{random} node among inactive nodes.
Figure~\ref{fig:results1} illustrates the empirical means of the expected utility $\tilde{f}$ as well as the $\pm 1$ standard deviation intervals over $100$ runs.
The adaptive greedy strategy significantly outperforms the other heuristic strategies in all cases.
Our results illustrate the empirical superiority of the greedy strategy to tackle the adaptive IM problem with myopic feedback, w.r.t. more common metrics from graph theory.
Without surprise, the random baseline is by far the worst strategy, while the performances of adaptive \textit{degree} and adaptive \textit{centrality} strategies seem to vary according to the networks.

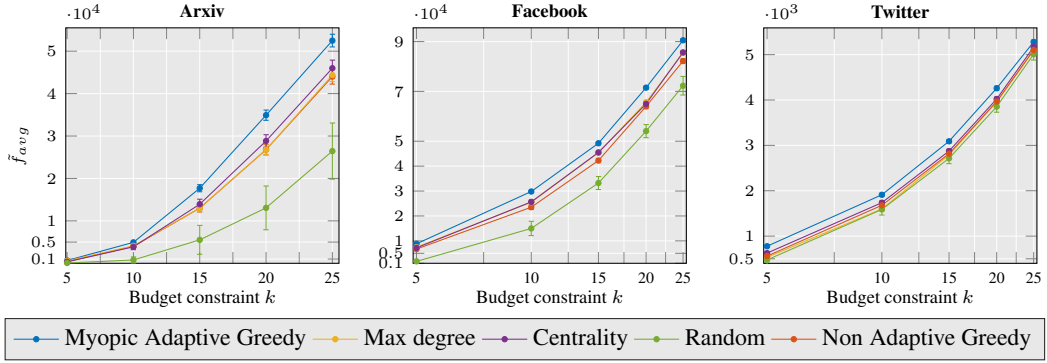
\begin{figure*}[t!]
  \centering
%
%
%
\begin{tikzpicture}

\begin{axis}[%
width=0.26\columnwidth,
scale only axis,
xmin=0.98,
xmax=5.1,
xtick={1,1.5,2,2.5,3,3.5,4,4.5,5},
xticklabels={{5},{},{10},{},{15},{},{20},{},{25}},
every x tick label/.append style={font=\tiny},
xlabel style={font=\color{white!15!black}},
xlabel style={yshift=0.3cm,font=\scriptsize},
xlabel={Budget constraint $k$},
ymin=0,
ymax=55501,
ytick={1000,5000,10000,20000,30000,40000,50000},
ylabel style={font=\color{white!15!black}},
ylabel={$\tilde{f}_{avg}$},
ylabel style={yshift=-0.6cm,font=\tiny},
every y tick label/.append style={font=\tiny},
axis background/.style={fill=white!91!black},
xmajorgrids,
ymajorgrids,
title={\bf Arxiv},
title style={yshift=-0.2cm,font=\scriptsize},
grid style={white, opacity=0.9},
legend style={legend cell align=left, align=left, draw=white!15!black, fill=white!91!black}
]
\addplot [color=mycolor1, mark size=1.pt, mark=*, mark options={solid, mycolor1}]
 plot [error bars/.cd, y dir = both, y explicit]
 table[row sep=crcr, y error plus index=2, y error minus index=3]{%
1	665	58	58\\
2	4935	302	302\\
3	17698	840	840\\
4	34913	1208	1208\\
5	52491	1492	1492\\
};

\addplot [color=mycolor2, mark size=1.pt, mark=*, mark options={solid, mycolor2}]
 plot [error bars/.cd, y dir = both, y explicit]
 table[row sep=crcr, y error plus index=2, y error minus index=3]{%
1	370	70	70\\
2	3999	591	591\\
3	13004	912	912\\
4	26785	1242	1242\\
5	44010	1774	1774\\
};

\addplot [color=mycolor3, mark size=1.pt, mark=*, mark options={solid, mycolor3}]
 plot [error bars/.cd, y dir = both, y explicit]
 table[row sep=crcr, y error plus index=2, y error minus index=3]{%
1	578	54	54\\
2	4113	306	306\\
3	13029	820	820\\
4	26777	1107	1107\\
5	44367	1521	1521\\
};

\addplot [color=mycolor4, mark size=1.pt, mark=*, mark options={solid, mycolor4}]
 plot [error bars/.cd, y dir = both, y explicit]
 table[row sep=crcr, y error plus index=2, y error minus index=3]{%
1	359	97	97\\
2	3878	668	668\\
3	13923	1184	1184\\
4	28832	1473	1473\\
5	45981	1911	1911\\
};

\addplot [color=mycolor5, mark size=1.pt, mark=*, mark options={solid, mycolor5}]
 plot [error bars/.cd, y dir = both, y explicit]
 table[row sep=crcr, y error plus index=2, y error minus index=3]{%
1	70	164	164\\
2	782	721	721\\
3	5525	3409	3409\\
4	13063	5157	5157\\
5	26449	6621	6621\\
};

\end{axis}
\end{tikzpicture}%
      \hfill
%
%
%
\begin{tikzpicture}

\begin{semilogxaxis}[%
width=.26\columnwidth,
scale only axis,
xmin=0.98,
xmax=5.1,
xtick={1,1.5,2,2.5,3,3.5,4,4.5,5},
xticklabels={{5},{},{10},{},{15},{},{20},{},{25}},
every x tick label/.append style={font=\tiny},
xlabel style={font=\color{white!15!black}},
xlabel={Budget constraint $k$},
xlabel style={yshift=0.3cm,font=\scriptsize},
ymin=1000,
ymax=95527,
ytick={1000,5000,10000,20000,30000,40000,50000,70000,90000},
ylabel style={font=\color{white!15!black}},
every y tick label/.append style={font=\tiny},
ylabel style={yshift=-0.6cm,font=\tiny},
axis background/.style={fill=white!91!black},
xmajorgrids,
ymajorgrids,
yminorgrids,
title={\bf Facebook},
title style={yshift=-0.2cm,font=\scriptsize},
grid style={white, opacity=0.9}
]
\addplot [color=mycolor1, mark size=1.pt, mark=*, mark options={solid, mycolor1}]
 plot [error bars/.cd, y dir = both, y explicit]
 table[row sep=crcr, y error plus index=2, y error minus index=3]{%
1	8821	302	302\\
2	29807	552	552\\
3	49181	601	601\\
4	71512	632	632\\	
5	90517	624	624\\
};

\addplot [color=mycolor2, mark size=1.pt, mark=*, mark options={solid, mycolor2}]
 plot [error bars/.cd, y dir = both, y explicit]
 table[row sep=crcr, y error plus index=2, y error minus index=3]{%
1	6705	480	480\\
2	23512	891	891\\
3	42252	801	801\\
4	64014	972	972\\
5	82185	891	891\\
};

\addplot [color=mycolor3, mark size=1.pt, mark=*, mark options={solid, mycolor3}]
 plot [error bars/.cd, y dir = both, y explicit]
 table[row sep=crcr, y error plus index=2, y error minus index=3]{%
1	7471	367	367\\
2	25486	560	560\\
3	45397	629	629\\
4	65630	654	654\\
5	85479	628	628\\
};

\addplot [color=mycolor4, mark size=1.pt, mark=*, mark options={solid, mycolor4}]
 plot [error bars/.cd, y dir = both, y explicit]
 table[row sep=crcr, y error plus index=2, y error minus index=3]{%
1	7189	382	382\\
2	25655	721	721\\
3	45499	661	661\\
4	65001	813	813\\
5	85652	629	629\\
};

\addplot [color=mycolor5, mark size=1.pt, mark=*, mark options={solid, mycolor5}]
 plot [error bars/.cd, y dir = both, y explicit]
 table[row sep=crcr, y error plus index=2, y error minus index=3]{%
1	1759	795	795\\
2	14991	2869	2869\\
3	33219	2615	2615\\
4	54055	2629	2629\\
5	72277	3729	3729\\
};

\end{semilogxaxis}
\end{tikzpicture}%
   \hfill
%
%
%
\begin{tikzpicture}

\begin{semilogxaxis}[%
width=.26\columnwidth,
scale only axis,
scaled y ticks=base 10:-3,
xmin=0.98,
xmax=5.1,
xtick={1,1.5,2,2.5,3,3.5,4,4.5,5},
xticklabels={{5},{},{10},{},{15},{},{20},{},{25}},
every x tick label/.append style={font=\tiny},
xlabel style={font=\color{white!15!black}},
xlabel style={yshift=0.3cm,font=\scriptsize},
xlabel={Budget constraint $k$},
ymin=400,
ymax=5590,
ytick={500,1e3,2e3,3e3,4e+03,5e+03},
ylabel style={font=\color{white!15!black}},
ylabel style={yshift=-0.6cm,font=\tiny},
every y tick label/.append style={font=\tiny},
axis background/.style={fill=white!91!black},
xmajorgrids,
ymajorgrids,
yminorgrids,
title={\bf Twitter},
title style={yshift=-0.2cm,font=\scriptsize},
grid style={white, opacity=0.9},
legend style={legend cell align=left, align=left, draw=white!15!black, fill=white!91!black, font=\footnotesize},
legend columns=-1,
legend entries={Myopic Adaptive Greedy,Max degree,Centrality, Random, Non Adaptive Greedy},
legend to name=leg:offline
]
\addplot [color=mycolor1, mark size=1.pt, mark=*, mark options={solid, mycolor1}]
 plot [error bars/.cd, y dir = both, y explicit]
 table[row sep=crcr, y error plus index=2, y error minus index=3]{%
1	777	29	29\\
2	1911	36	36\\
3	3090	39	39\\
4	4259	45	45\\
5	5280	41	41\\
};

\addplot [color=mycolor3, mark size=1.pt, mark=*, mark options={solid, mycolor3}]
 plot [error bars/.cd, y dir = both, y explicit]
 table[row sep=crcr, y error plus index=2, y error minus index=3]{%
1	532	36	36\\
2	1599	48	48\\
3	2790	54	54\\
4	3960	53	53\\
5	5092	48	48\\
};

\addplot [color=mycolor4, mark size=1.pt, mark=*, mark options={solid, mycolor4}]
 plot [error bars/.cd, y dir = both, y explicit]
 table[row sep=crcr, y error plus index=2, y error minus index=3]{%
1	626	30	30\\
2	1737	37	37\\
3	2875	41	41\\
4	4022	44	44\\
5	5168	40	40\\
};

\addplot [color=mycolor5, mark size=1.pt, mark=*, mark options={solid, mycolor5}]
 plot [error bars/.cd, y dir = both, y explicit]
 table[row sep=crcr, y error plus index=2, y error minus index=3]{%
1	475	118	118\\
2	1588	123	123\\
3	2714	117	117\\
4	3854	123	123\\
5	5009	129	129\\
};

\addplot [color=mycolor2, mark size=1.pt, mark=*, mark options={solid, mycolor2}]
 plot [error bars/.cd, y dir = both, y explicit]
 table[row sep=crcr, y error plus index=2, y error minus index=3]{%
1	560	49	49\\
2	1679	50	50\\
3	2817	54	54\\
4	3967	53	53\\
5	5102	58	58\\
};

\end{semilogxaxis}
\end{tikzpicture}%
  \ref{leg:offline}
  \caption{Expected cumulative number of active nodes ($\tilde{f}_{avg}$) vs. number of seeds for real world networks.}
  \label{fig:results1}
\end{figure*}

 \noindent\textbf{Adaptive greedy Vs Non-adaptive greedy}
Comparisons have also been made with a non-adaptive standard greedy strategy \citep{kdd/KempeKT03}.
This policy chooses the $k$ seed nodes in advance, at $t=1$, and activate each one of them sequentially (one at each time step).
Based on our experiments (see Fig.~\ref{fig:results1}), the adaptive greedy strategy provides larger influence spreads than the non-adaptive greedy.
It becomes apparent that adaptivity is more profitable, as we gradually gain more knowledge about the truth network.
The performance of the non-adaptive greedy strategy is sometimes worse even when it is compared with that of the adaptive \emph{degree} or \emph{centrality} strategies.
Overall, the results validate our initial claim that the performance of the proposed myopic adaptive greedy policy will be at least as good as that of the non-adaptive greedy policy. 

\noindent\textbf{Impact of network's structure on performance}
Another main insight from our empirical study is that the network's structure strongly impacts the performance of algorithms.
While the superiority of the adaptive greedy strategy is clear on Arxiv and Facebook data, differences between strategies are less obvious on Twitter's network.
It highlights that increasing the edges/nodes ratio of the network decreases the global advantage of the adaptive greedy policy on other strategies. 
Actually, the IM problem itself is less relevant when the network becomes very dense, as all nodes have a quite similar influence power.
As a consequence, it is not surprising to obtain smaller differences between strategies on Twitter.
Since this network is very dense, even the random baseline manages to return good spreads.

\section{Conclusions}

We presented the \emph{myopic adaptive greedy strategy} for the adaptive influence maximization task. 
It is the first time that a policy like this one offers provable approximation guarantees under an IM diffusion model with myopic feedback.
Actually, it is achieved by maximizing an alternative utility function that considers the cumulative number of active nodes over time instead of the total number of the active nodes at the end of the diffusion process.
Our experiments illustrated the empirical superiority of the proposed strategy over more common approaches from graph theory.
Our analysis also pointed out how the graph's density strongly impacts the performance of algorithms.

Several interesting issues remain open for future work. So far, we considered that the influence graph was fully known, which may be a strong assumption in practice. We intend to relax this assumption, studying problems where influence probabilities must be adaptively learned in order to maximize influence.
Last but not least, we plan to examine an even more realistic version of the \emph{modified} IC model.
In that case, the influence probabilities between an active node and its inactive neighbors will be decreased by a predefined factor right after each failure of the first node to influence the other ones.

\newpage

\section*{A. Adaptive setting leads to higher spreads}

\begin{claim}
  The adaptive setting leads to higher spreads compared to the non-adaptive one, since we gradually gain more knowledge about the ground truth influence graph.
\end{claim}
To defend our claim, we give a simple example. Consider the network shown in Fig.~\ref{fig:toyexample}(a)  with influence probabilities $p_{vu} = 0.9$ and $p_{vw} = 0.1$.
Let $k = 2$ (seed nodes - our budget). The non-adaptive greedy algorithm will select as seed nodes the $v$ ($t=1$) and $w$ ($t=2$). Nevertheless, based on the true world (see Fig.~\ref{fig:toyexample}(b)), we observe that nodes $v$ and $w$ are active at time $t=2$. Hence, we will infer that the edges $(v,w)$ and $(v,u)$ are live and dead, respectively. Therefore, the non-adaptive strategy will lead to a reward equal to $2$, as only nodes $v$ and $w$ will be activated finally, but not $u$. Roughly speaking, we are going to make an offer at an already influenced user. On the other hand, the adaptive myopic strategy will first choose the node $v$ and then will observe the status of the outgoing edges of node $v$. In other words, he will observe that $v$ managed to influence $w$ but not $u$. Hence, he will choose node $u$ as the second seed node, since it is the only one which is not activated at this point. This returns a reward equal to $3$, which is higher than that returned by the non-adaptive policy, since all nodes are finally activated.

\begin{figure}[h] 
  \begin{center}
    \subfigure[]{
      \begin{tikzpicture}[->, >=stealth', auto, semithick, node distance=2.cm]
\tikzstyle{every state}=[fill=white,draw=black,thick,text=black,scale=1.]
\node[state] (0) {$v$};
\node[state] (1)[right of=0, xshift=.3cm]   {$u$};
\node[state] (2)[below of=1]   {$w$};

\path
(0) edge[left,above]    node{$0.9$}      (1)
(0) edge[left,below]    node{$0.1$}      (2);
\end{tikzpicture}
    } \hspace{4em}
    \subfigure[]{
      \begin{tikzpicture}[->, >=stealth', auto, semithick, node distance=2.cm]
\tikzstyle{every state}=[fill=white,draw=black,thick,text=black,scale=1.]

\node[state, fill=gray!30]    (10)[right of=1] {$v$};
\node[state] (11)[right of=10, xshift=.3cm] {$u$};
\node[state, fill=gray!30]    (12)[below of=11] {$w$};

\path
(10) edge[left,above, dashed]    node{$0$}      (11)
(10) edge[left,below]    node{$1$}      (12);

\end{tikzpicture}
    }
    \caption{(a) Graph network (b) True world at time $t=2$}
    \label{fig:toyexample}
  \end{center}
\end{figure}
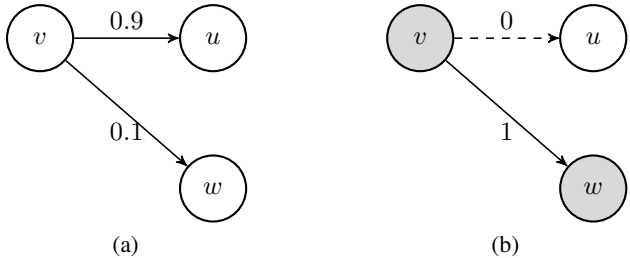

\section*{B. Networks Description}

Experiments have been conducted on three networks obtained from the Stanford's SNAP database \cite{snapnets}. 
The first one is a small graph that corresponds to an \textit{ego network from Twitter}.
Actually, the dataset is a subset - a ``circle'' - from the list of social circles from Twitter, crawled from public sources.
The graph consists of \textit{$228$ nodes} and \textit{$9,938$ edges}.

The second one is a \textit{social network extracted from Facebook}. Data were anonymously collected from survey participants using the Facebook app.
The graph is undirected, and has \textit{$4,039$ nodes} and \textit{$88,234$ edges}.

The third graph is that of the \textit{Arxiv General Relativity and Quantum Cosmology collaboration network}.
In this graph, we have an undirected edge from $i$ to $j$, if author $i$ co-authored a ArXiv paper with author $j$ (between $1993$ and $2003$).
This graph has \textit{$5,242$ nodes} and \textit{$28,980$ edges}.

Table~\ref{table1} summarizes a number of useful statistics about the aforementioned networks.
\textit{Mean degree} is the mean number of edges \textit{exiting} nodes. 
A.P.L. stands for \textit{Average Path Length}, which is the average number of nodes in the shortest path between two nodes of the graph. Moreover,
the \textit{diameter} of a graph is the length of the longest shortest path between two nodes.
\begin{table*}[h]
  \small
  \begin{center}
    \begin{tabular}{c|ccccccc}
      \textbf{Network} & \textbf{Nodes} & \textbf{Edges} & \textbf{Mean degree} & \textbf{Max degree} & \textbf{A.P.L.} & \textbf{Diameter} & \textbf{Type} \\
      \hline 
      Twitter & $228$ & $9,938$ & $43.6$ & $125$ & $2.1$ & $6$ & Directed \\
      ArXiv GR-QC & $5,242$ & $28,980$ & $11.1$ & $162$ & $6.1$ & $17$ & Undirected \\
      Facebook & $4,039$ & $88,234$ & $43.7$ & $1,045$ & $3.7$ & $8$ & Undirected \\
    \end{tabular}
    \caption{Statistics of the real-world networks used through the experimental analysis.}     \label{table1}
  \end{center}
\end{table*}

\section*{C. Adaptive Greedy Myopic Policy and Alternative Heuristics}

The performance of the proposed \emph{myopic adaptive greedy strategy} has been compared with that of the next four \textit{alternative adaptive heuristics}.
\begin{itemize}[leftmargin=*,noitemsep,nolistsep]
\item \textbf{Degree:}  The node with the \textit{highest degree} (i.e., the node with the highest number of outgoing edges) has been chosen as a seed node at each time;
\item \textbf{Centrality:} The node(s) with the \textit{highest centrality measure} among the inactive nodes has been selected as a seed node at each time step. In our analysis, we adopted the \textit{betweenness centrality} measure, which is equal to the number of shortest paths from all nodes to all others that pass through a node;
\item \textbf{Random:} Selects randomly an inactive node as a seed node at each time step;
\item \textbf{Non-adaptive:} The $k$ seed nodes have been selected in advance by using the standard greedy algorithm \cite{kdd/KempeKT03}. Then, we activate each of them sequentially, at each time step, starting from the one with the maximum expected marginal gain.
\end{itemize}

Finally, it should be stressed that in the case of the \emph{modified myopic feedback model}, we chose to implement the improved \textit{accelerated version} of the \textit{adaptive greedy strategy} \cite{Golovin:2011}, for computational reasons.
The algorithm is based on so-called \textit{lazy evaluations}, i.e. on a clever use of the adaptive submodularity inequality to significantly reduce running times in practice by diminishing the number of nodes on which Monte Carlo simulations should be performed. 
The pseudocode and the justification of this acccelerated adaptive greedy algorithm are reported in \cite{Golovin:2011}.

\bibliography{misc}
\bibliographystyle{plainnat}

\end{document}